\documentclass[a4paper]{article}

\usepackage[english]{babel}
\usepackage[utf8]{inputenc}
\usepackage{amsmath}
\usepackage{amssymb}
\usepackage{amsthm}
\newtheorem{theorem}{Theorem}[section]
\newtheorem{lemma}{Lemma}[section]
\newtheorem{corollary}{Corollary}[section]
\usepackage{graphicx}
\usepackage[colorinlistoftodos]{todonotes}

\title{An Important Corollary for the Fast Solution of Dynamic Maximal Clique Enumeration Problems}

\author{Neal Lawton}

\date{\today}

\begin{document}
\maketitle

\begin{abstract}
In this paper we modify an algorithm for updating a maximal clique enumeration after an edge insertion to provide an algorithm that runs in linear time with respect to the number of cliques containing one of the edge's endpoints, whereas existing algorithms take quadratic time.
\end{abstract}

\section{Introduction}

\subsection{Terminology and Notation}

For a graph $G(V,E)$, a \textit{clique} $C$ in $G$ is a subset of vertices that are all adjacent to each other in $G$, i.e., $C$ is a complete subgraph of $G$. A clique $C$ in $G$ is \textit{maximal} if no other clique in $G$ contains $C$. An edge with endpoints $u, v \in V$ is denoted $uv$. For convenience, when we refer to the \textit{neighborhood} of a vertex $u \in V$, we actually mean the closed neighborhood, i.e., the set $\{v : v = u \text{ or } uv \in E \}$, and denote this $N(u)$.

\subsection{Problem Statement}

The solution to a Maximal Clique Enumeration (MCE) Problem for a graph $G$ is a list of all maximal cliques in $G$. In dynamic MCE Problems, our task is to update the list of maximal cliques after inserting or deleting an edge. In the case of inserting an edge, our task is to list all maximal cliques in $G(V, E \cup \{uv\})$ given the list of all maximal cliques in $G(V,E)$.

\subsection{Motivation}

The MCE problem for a static graph has many important applications, such as in solving graph coloring problems: since every pair of vertices in a clique must be assigned different colors, a list of maximal cliques allows us to prune the search tree immensely. Graph coloring, in turn, is used in many scheduling problems. For example, given a list of time intervals during which various flights will need to use a gate, we can construct a graph whose vertices represent flights and whose edges connect vertices whose corresponding flights require a gate at the same time. The chromatic number of this graph is the number of gates required to service all the flights. However, unexpected flight delays induce changes to the graph throughout the day, requiring the solution of another slightly different graph coloring problem and thus another slightly different maximal clique enumeration problem.

The dynamic MCE problem also finds application in computational topolgoy. Given a point cloud $V$ in a metric space and a radius $\epsilon$, we can construct a graph analogous to the Vietoris-Rips Complex $R_{\epsilon}(V)$, where vertices represent points in the point cloud and where two vertices are connected by an edge if the corresponding points in the point cloud lie within distance $\epsilon$ of each other. For $\epsilon = 1$ and $V \subset \mathbb{R}^2$, the graph is a unit disk graph. In computational topology, we can compute the homology for a single graph by enumerating its maximal cliques and using the algorithm proposed by \cite{zomorodian}. Characterizing the topology of the point cloud by the homology of the VR Complex for a single $\epsilon$ is a bad idea, since noise can produce false topological features. What we really want to know is which topological features persist over a wide range of $\epsilon$. Calculting and sorting all pairwise distances of points in $V$ induces a sequence of graphs $\{G_n\}$ where the edges of the $n$th graph are the first $n$ edges of the sorted list of distances. Listing all maximal cliques for each of these graphs is a dynamic MCE Problem, the solution of which is used to construct a \textit{persistence barcode} that illustrates the emergence and disappearance of topological features of the VR Complex as $\epsilon$ varies.

\section{Existing Method}

For completeness, we'll begin by proving the correctness of the existing method for updating the list of maximal cliques after an edge insertion.

\begin{lemma}
\label{Symmetry}
For a graph $G(V,E)$, if $C$ is a clique in $G(V,E \cup \{uv\})$, then $C \setminus u$ (and, by symmetry, $C \setminus v$) is a clique in $G(V,E)$.
\end{lemma}

\begin{proof}
$C \setminus u \subset C$, so $C \setminus u$ is a clique in $G(V, E \cup \{uv\})$, i.e., $s,t \in C \setminus u \implies st \in E \cup \{uv\}$. Since $u \not \in C \setminus u$, we must have $st \in E$, i.e., $C$ is a clique in $G(V,E)$.
\end {proof}

The following lemma says that if $C$ is a maximal clique in $G(V,E)$ but not in $G(V, E \cup \{uv\})$, then $C \subset N(u) \cup N(v)$.

\begin{lemma}
\label{Disappear}
For a graph $G(V,E)$ and $C \subset V$, if $u,v \not \in C$, then $C$ is a maximal clique in $G(V,E)$ $\implies$ $C$ is a maximal clique in $G(V,E\cup \{uv\})$.
\end{lemma}

\begin{proof}
Since $C$ is a clique in $G(V,E)$, $s,t \in C \implies st \in E \implies st \in E \cup \{uv\}$, so $C$ is a clique in $G(V,E \cup \{uv\})$. Let $B \subset V$ be a clique in $G(V, E \cup \{uv\})$ such that $C \subset B$. Since $B$ is a clique in $G(V, E\cup \{uv\})$, $s,t \in B \implies st \in E \cup \{uv\}$. However, $u,v \not \in C \implies u \not \in B$ or $v \not \in B$, since $u,v \in B \implies B \setminus u \supset C$ strictly and $B \setminus u$ is a clique in $G(V,E)$, which contradicts the maximality of $C$ in $G(V,E)$. So $s,t \in B \implies st \in E$, i.e., $B$ is a clique in $G(V,E)$. Since $C$ is maximal in $G(V,E)$, we must have $B = C$. Since $B$ was arbitrary, $C$ is maximal in $G(V, E \cup \{uv\})$.
\end{proof}

The following lemma says that if $C$ is a maximal clique in $G(V, E \cup \{uv\})$ but not in $G(V, E)$, then $C \subset N(u) \cap N(v)$.

\begin{lemma}
\label{Emerge}
For a graph $G(V,E)$ and $C \subset V$, if $u \not \in C$ or $v \not \in C$, then $C$ is a maximal clique in $G(V,E \cup \{uv\})$ $\implies$ $C$ is a maximal clique in $G(V,E)$.
\end{lemma}

\begin{proof}
Without loss of generality, $u \not \in C$. Since $C$ is a clique in $G(V, E \cup \{uv\})$, $s,t \in C \implies st \in E \cup \{uv\}$. Since $u \not \in C$, we must have $st \in E$, i.e., $C$ is a clique in $G(V,E)$. Let $B \subset V$ be a clique in $G(V,E)$ such that $C \subset B$. Then $s,t \in B \implies st \in E \implies st \in E \cup \{uv\}$, so B is a clique in $G(V, E \cup \{uv\})$. Since $C$ is maximal in $G(V, E \cup \{uv\})$, we must have $B = C$. Since $B$ was arbitrary, $C$ is maximal in $G(V,E)$.
\end{proof}

The following is the main theorem that justifies the existing method. It's also the source of the bottleneck in the existing method and the target of our improvement.

\begin{theorem}
\label{ExistingMethod}
For a graph $G(V,E)$ and $C \subset V$, if $u,v \in C$ and $C$ is a maximal clique in $G(V, E \cup \{uv\})$, then there exist maximal cliques $C_u, C_v \subset V$ in $G(V,E)$ such that $u \in C_u$, $v \in C_v$, and $C$ = $(C_u \cap C_v) \cup \{u,v\}$.
\end{theorem}

\begin{proof}
$C \setminus v$ and $C \setminus u$ are cliques in $G(V,E)$, and thus each lie within some maximal clique in $G(V,E)$, say $C_u$ and $C_v$, respectively. $u \in C \setminus v \implies u \in C_u$, and $v \in C \setminus u \implies v \in C_v$. Since $C \setminus v \subset C_u$ and $C \setminus u \subset C_v$, we have $(C \setminus v) \cap (C \setminus u) \subset C_u \cap C_v$. So $C = ((C \setminus v) \cap (C \setminus u)) \cup \{u, v\} \subset (C_u \cap C_v) \cup \{u,v\}$. However, this inclusion cannot be strict, since $(C_u \cap C_v) \cup \{u,v\}$ is a clique in $G(V, E \cup \{u,v\})$ and $C$ is maximal. Thus $C = (C_u \cap C_v) \cup \{u,v\}$.
\end{proof}

The previous theorems are enough to justify the existing method. Lemma ~\ref{Emerge} tells us that new maximal cliques emerge only in $N(u) \cap N(v)$, and Lemma ~\ref{Disappear} says that cliques lose their maximality only if they lie in $N(u) \cup N(v)$. Together, they say that everything outside $N(u) \cup N(v)$ remains the same. Theorem ~\ref{ExistingMethod} gives us a method for finding the new maximal cliques: for each maximal clique $C_u$ in $G$ containing $u$ and each maximal clique $C_v$ in G containing $v$, generate $(C_u \cap C_v) \cup \{u,v\}$ and mark it as a candidate for maximality. Then we test each candidate for maximality and add the candidate to the enumeration only if it passes the test. Theorem 5.3 in \cite{hendrix} gives a method for testing a candidate's maximality in $O(|N(u) \cap N(v)|^2)$ time. If there's $|\{C_u\}|$ maximal cliques containing $u$ and $|\{C_v\}|$ maximal cliques containing $v$, then checking each candidate against the others takes $O(|\{C_u\}| |\{C_v\}| |N(u) \cap N(v)|^2)$ time.

\section{Main Result}

The following Corollary justifies our improvement to the existing method.

\begin{corollary}
\label{ProposedMethod}
For a graph $G(V,E)$ and $C \subset V$, if $u,v \in C$ and $C$ is a maximal clique in $G(V,E \cup \{u,v\})$, then there exist maximal cliques $C_u, C_v \subset V$ in $G(V,E)$ such that $u \in C_u$, $v \in C_v$, and $C = (C_u \cap N(v)) \cup \{u,v\} = (N(u) \cap C_v) \cup \{u,v\}$.
\end{corollary}

\begin{proof}
Let $C_u, C_v$ be the maximal cliques in $G(V,E)$ guaranteed by Theorem ~\ref{ExistingMethod}. $C = (C_u \cap C_v) \cup \{u,v\} \subset (C_u \cap N(v)) \cup \{u,v\}, (N(u) \cap C_v) \cup \{u,v\}$. Since every subset of a clique is a clique, $(C_u \cap N(v)) \cup \{u,v\}$ and $(N(u) \cap C_v) \cup \{u,v\}$ are cliques in $G(V, E \cup \{u,v\})$. Since $C$ is a maximal clique in $G(V, E \cup \{u,v\})$, the inclusion must actually be equality.
\end{proof}

Now that we no longer need to generate pairwise intersections, our list of candidate cliques is $|\{C_u\}|$ (or $|\{C_v\}|$ if we generate candidates from $\{C_v\}$) and checking the candidates for maximality takes $O(|\{C_u\}| |N(u) \cap N(v)|^2)$ time. This is significant since $|\{C_v\}|$ can grow exponentially with the number of vertices in the graph $|V|$. Note that generating candidates from $\{C_u\}$ and $\{C_v\}$ is redundant, since in the proof $(C_u \cap N(v)) \cup \{u,v\} = (N(u) \cap C_v) \cup \{u,v\}$. This means we now only require either $\{C_u\}$ or $\{C_v\}$. If we have both lists, we could save time by choosing to generate candidates using the shorter list only.

\section{Maximal k-Cliques}

A \textit{maximal $k$-clique} $C \subset V$ is a subset of vertices that is either a clique in $G$ of size $k$ or a maximal clique in $G$ of size $< k$. The number of $k$-cliques in a graph grows in polynomial time: there are at most $\sum_{i=1}^k \binom{|V|}{i} \in O(|V|^k)$ maximal $k$-cliques in a graph. The task of the maximal $k$-clique enumeration problem is to enumerate all maximal $k$-cliques in a graph, and is much easier than the classic MCE problem. Similarly, the task of the dynamic maximal $k$-clique enumeration problem is to update the maximal $k$-clique enumeration after inserting or deleting an edge. Researchers are sometimes willing to sacrifice information about larger cliques in exchange for better runtimes. Returning to computational topology, surface reconstruction from a point cloud only requires knowledge about 3-cliques. Fortunately, Corollary ~\ref{ProposedMethod} can also be used to accelerate the dynamic maximal $k$-clique enumeration problem in the case of inserting an edge.

The proofs of Lemma ~\ref{Symmetry}, Lemma ~\ref{Disappear}, and Lemma ~\ref{Emerge} are identical to those used to prove equivalent statements for maximal $k$-cliques. However, adjusting Theorem ~\ref{ExistingMethod} for $k$-cliques requires some attention to detail:

\begin{theorem}
\label{ExistingKClique}
For a graph $G(V,E)$ and $C \subset V$, if $u,v \in C$ and $C$ is a maximal k-clique in $G(V, E \cup \{uv\})$, then either $|C| < k$ and there exist maximal k-cliques $C_u, C_v \subset V$ in $G(V,E)$ such that $u \in C_u$, $v \in C_v$, and $C$ = $(C_u \cap C_v) \cup \{u,v\}$, or $|C| = k$ and $C \subset (C_u \cap C_v) \cup \{u,v\}$ and $|(C_u \cap C_v) \cup \{u,v\}| \leq k+1$.
\end{theorem}

\begin{proof}
$C \setminus v$ and $C \setminus u$ are cliques in $G(V,E)$ of size $\leq k$, and thus each lie within some maximal $k$-clique in $G(V,E)$, say $C_u$ and $C_v$, respectively. $u \in C \setminus v \implies u \in C_u$, and $v \in C \setminus u \implies v \in C_v$. Since $C \setminus v \subset C_u$ and $C \setminus u \subset C_v$, we have $(C \setminus v) \cap (C \setminus u) \subset C_u \cap C_v$. So $C = ((C \setminus v) \cap (C \setminus u)) \cup \{u, v\} \subset (C_u \cap C_v) \cup \{u,v\}$. If $|C| < k$, then $C$ is maximal in $G$ and this inclusion cannot be strict, since $(C_u \cap C_v) \cup \{u,v\}$ is a clique in $G(V, E \cup \{u,v\})$ and $C$ is maximal; in this case, $C = (C_u \cap C_v) \cup \{u,v\}$. Else, $|C| = k$. $|C_u|, |C_v| \leq k$, so $|C_u \cap C_v| \leq k$. If  $C_u = C_v$, as may be the case when $uv \in E$, then $(C_u \cap C_v) \cup \{u,v\} = C_u \cap C_v$ so $|(C_u \cap C_v) \cup \{u,v\}| \leq k$. If $C_u \neq C_v$, then $|C_u \cap C_v| \leq k-1$ so $|(C_u \cap C_v) \cup \{u,v\}| \leq |(C_u \cap C_v)| + |\{u,v\}| = k+1$.
\end{proof}

So updating the list of maximal $k$-cliques is similar to updating the list of maximal cliques, except that now in the case where $k \leq |(C_u \cap C_v) \cup \{u,v\}|$, we don't have to run any maximality test; every size-$k$ subset of the candidate becomes a maximal $k$-clique.

Corollary ~\ref{ProposedMethod} lets us generate a linear number of candidates for $k$-cliques, too.

\begin{corollary}
\label{ProposedKClique}
For a graph $G(V,E)$ and $C \subset V$, if $u,v \in C$ and $C$ is a maximal $k$-clique in $G(V,E \cup \{u,v\})$, then there exist maximal $k$-cliques $C_u, C_v \subset V$ in $G(V,E)$ such that $u \in C_u$, $v \in C_v$, and $C = (C_u \cap N(v)) \cup \{u,v\} = (N(u) \cap C_v) \cup \{u,v\}$ or $|C| = k$ and $C \subset (C_u \cap N(v)) \cup \{u,v\}, (N(u) \cap C_v) \cup \{u,v\}$.
\end{corollary}

\begin{proof}
Let $C_u, C_v$ be the maximal $k$-cliques in $G(V,E)$ guaranteed by Theorem ~\ref{ExistingKClique}. $C \subset (C_u \cap C_v) \cup \{u,v\} \subset (C_u \cap N(v)) \cup \{u,v\}, (N(u) \cap C_v) \cup \{u,v\}$. Since every subset of a clique is a clique, $(C_u \cap N(v)) \cup \{u,v\}$ and $(N(u) \cap C_v) \cup \{u,v\}$ are cliques in $G(V, E \cup \{u,v\})$. If $|C| < k$, then $C$ is a maximal clique in $G(V, E \cup \{u,v\})$, so the inclusion must actually be equality.
\end{proof}

We follow the same procedure for testing candidates as before: if $|(C_u \cap N(v)) \cup \{u,v\}| < k$, use the maximality test from \cite{hendrix}; else, every size-$k$ subset becomes a maximal $k$-clique. As in the application of Corollary ~\ref{ProposedMethod} to the traditional MCE setting, we only need either $\{C_u\}$ or $\{C_v\}$ to correctly generate the list of $k$-cliques for the updated graph since $C \subset (C_u \cap N(v)) \cup \{u,v\}$ and $C \subset (N(u) \cap C_v) \cup \{u,v\}$.

\section{Parallelism}

Lemmas ~\ref{Disappear} and ~\ref{Emerge} tell us that the list of maximal cliques doesn't change outside of $N(u) \cup N(v)$. If $E_n \subset E_{n+1} = E_n \cup \{u_1 v_1\} \subset E_{n+2} = E_{n+1} \cup \{u_2 v_2\}$, and $(N(u_2) \cup N(v_2)) \cap \{u_1 v_1 \} = \emptyset$, then inserting $u_1 v_1$ affects neither $\{C_{u_2}\}$ nor $\{C_{v_2}\}$, so the updates from $G(V,E_n)$ to $G(V,E_{n+1})$ and from $G(V,E_{n+1})$ to $G(V,E_{n+2})$  are independent and can run in parallel, even with the existing method. Since we only need either $\{C_u\}$ or $\{C_v\}$ to use the proposed method, we have more opportunities for parallelism: if only $N(u_2) \cap \{u_1 v_1\} = \emptyset$, then inserting $u_1 v_1$ doesn't affect $\{C_{u_2}\}$, so the updates can still run in parallel if we use the proposed method.

\end{document}